\documentclass[12pt,a4paper]{article}

\usepackage{fullpage}

\usepackage{amssymb}
\usepackage{amsmath}
\usepackage{amsthm}
\usepackage[round,longnamesfirst]{natbib}
\usepackage{url}
\usepackage{graphicx}
\usepackage{multirow}
\usepackage{enumerate}
\usepackage{setspace}
\usepackage{booktabs,rotating}
\usepackage{multirow}
\usepackage[nolists]{endfloat}
\usepackage{endrotfloat}

\newcommand{\Rset}{\mathbb{R}}
\newcommand{\FF}{\mathcal{F}}
\newcommand{\GG}{\mathcal{G}}
\newcommand{\HH}{\mathcal{H}}
\newcommand{\OO}{\mathcal{O}}
\newcommand{\MM}{\mathcal{M}}
\newcommand{\Fb}{\mathbb{F}}

\newcommand{\1}{\mathbf{1}}
\newcommand{\Gn}{\mathbb{G}_n}
\newcommand{\G}{\mathbb{G}}

\newcommand{\Pn}{\mathbb{P}_n}
\newcommand{\dd}{\mathrm{d}}

\newtheorem{prop}{Proposition}
\newtheorem{lem}{Lemma}

\title{Goodness-of-fit testing based on a weighted bootstrap: A fast large-sample alternative to the parametric bootstrap}

\author{
   Ivan Kojadinovic\\ 
   \small{Laboratoire de math\'ematiques et applications, UMR CNRS 5142} \\
   \small{Universit\'e de Pau et des Pays de l'Adour} \\
   \small{B.P. 1155, 64013 Pau Cedex, France} \\
   \small{\texttt{ivan.kojadinovic@univ-pau.fr}}
   \and
   Jun Yan\\ 
   \small{Department of Statistics} \\
   \small{University of Connecticut, 215 Glenbrook Rd. U-4120} \\
   \small{Storrs, CT 06269, USA} \\
   \small{\texttt{jun.yan@uconn.edu}}
 }

\date{}

\begin{document}
\maketitle    

\begin{abstract}
  The process comparing the empirical cumulative distribution function of the sample with a parametric estimate of the cumulative distribution function is known as the {\em empirical process with estimated parameters} and has been extensively employed in the literature for goodness-of-fit testing. 
The simplest way to carry out such goodness-of-fit tests, especially in a multivariate setting, is to use a {\em parametric bootstrap}. Although very easy to implement, the parametric bootstrap can become very computationally expensive as the sample size, the number of parameters, or the dimension of the data increase. An alternative resampling technique based on a fast {\em weighted bootstrap} is proposed in this paper, and is studied both theoretically and empirically. The outcome of this work is a generic and computationally efficient {\em multiplier} goodness-of-fit procedure that can be used as a large-sample alternative to the parametric bootstrap. In order to approximately determine how large the sample size needs to be for the parametric and weighted bootstraps to have roughly equivalent powers, extensive Monte Carlo experiments are carried out in dimension one, two and three, and for models containing up to nine parameters. The computational gains resulting from the use of the proposed multiplier goodness-of-fit procedure are illustrated on trivariate financial data. A by-product of this work is a fast large-sample goodness-of-fit procedure for the bivariate and trivariate $t$ distribution whose degrees of freedom are fixed.

\medskip


\end{abstract}


\section{Introduction}

Let $\MM = \{F_\theta : \theta \in \OO \}$ be a parametric family of cumulative distribution functions (c.d.f.s) on $\Rset^d$, where $\OO$ is an open subset of $\Rset^p$, for some integers $d \geq 1$ and $p \geq 1$. Given a sample $X_1,\dots,X_n$ of i.i.d.\ random vectors on $\Rset^d$ with common c.d.f.\ $F$, we are interested in testing
$$
H_0 : F \in \MM \qquad \mbox{against} \qquad H_1 : F \not \in \MM
$$
using statistics based on the empirical process
\begin{equation}
\label{test_process}
\Fb_n(x) = \sqrt{n} \{ F_n(x) - F_{\theta_n}(x) \}, \qquad x \in \Rset^d,
\end{equation}
where 
$$
F_n(x) = \frac{1}{n} \sum_{i=1}^n \1(X_i \leq x), \qquad x \in \Rset^d,
$$
is the empirical c.d.f.\ computed from the random sample $X_1,\dots,X_n$, and $F_{\theta_n}$ is a parametric estimator of $F$ computed under the null hypothesis that $F$ belongs to $\MM$, that is, under the assumption that there exists $\theta_0 \in \OO$ such that $F = F_{\theta_0}$. The latter parametric estimator of $F$ under $H_0$ is obtained from an estimator $\theta_n$ of $\theta_0$ based on $X_1,\dots,X_n$. The above problem has been extensively studied in the literature \citep[see e.g.][]{Dar55,KacKieWol55,Suk72,Dur73,Ste76,Khm81,Dur75} and is often referred to as the problem of {\em goodness of fit when parameters are estimated}. 

Two test statistics that are frequently used are the Cram\'er--von Mises statistic
\begin{equation}
\label{Sn}
S_n = n \int_{\Rset^d} \{ F_n(x) - F_{\theta_n}(x) \}^2 \dd F_{\theta_n}(x)
\end{equation}
and the Kolmogorov-Smirnov statistic
\begin{equation}
\label{Tn}
T_n = \sqrt{n} \sup_{x \in \Rset^d} | F_n(x) - F_{\theta_n}(x) |.
\end{equation}
Minor variations of these will also be considered in this work.

One important and enduring issue regarding tests of goodness of fit based on~(\ref{test_process}) concerns the computation of critical values or $p$-values for statistics derived from $\Fb_n$. Indeed, under classical regularity conditions that will be explicitly stated in the forthcoming section, the weak limit of the process given in~(\ref{test_process}) involves a drift term that typically makes goodness-of-fit tests based on $\Fb_n$ distribution-dependent. 

To solve this problem in the univariate case, \cite{Khm81} proposed to use the theory of martingales to transform $\Fb_n$ to an asymptotically distribution-free process. In the same context, \citet{Dur73,Dur75} investigated several approaches to compute approximate critical values for the Kolmogorov-Smirnov statistic. The martingale transform of \citet{Khm81} and \citet{Dur75}'s approach based on approximate boundary crossing probabilities are reviewed and compared in \cite{Par10} when $\MM$ is a location-scale or a scale-shape univariate family. For tests based on the Cram\'er--von Mises, the Anderson--Darling or the Watson statistics, \cite{Ste76} \citep[see also][]{Suk72,Ste74} used the fact that the asymptotic distributions of these statistics can be expressed as a weighted sum of $\chi_1^2$ variables and explained in detail how to compute the unknown weights when $\MM$ is the univariate normal or the exponential distribution.

A generic, very simple to implement resampling technique, that can be used to carry out tests of goodness of fit based on $\Fb_n$ in a general multivariate setting, is the so-called {\em parametric bootstrap}  \citep[see e.g.][]{Rom88,StuGonPre93,JogRao04,GenRem08}. To fix ideas, let $S_n$ be the test statistic and assume that it is a continuous functional of the empirical process $\Fb_n$. We shall thus write $S_n = \phi(\Fb_n)$. For some large integer $N$ and given an estimator $\theta_n$ of $\theta_0$ based on $X_1,\dots,X_n$, the parametric bootstrap consists of repeating the following steps for every $k \in \{1,\dots,N\}$:
\begin{enumerate}[(a)]
\item Generate a random sample $X_1^{(k)},\dots,X_n^{(k)}$ from c.d.f.\ $F_{\theta_n}$.
\item Let $F_n^{(k)}$ and $\theta_n^{(k)}$ stand for the versions of $F_n$ and $\theta_n$ estimated from the random sample $X_1^{(k)},\dots,X_n^{(k)}$.
\item Form an approximate realization of the test statistic under $H_0$ as $S_n^{(k)} = \phi(\Fb_n^{(k)})$, where $\Fb_n^{(k)} = \sqrt{n} (F_n^{(k)} - F_{\theta_n^{(k)}} )$.
\end{enumerate}
With the convention that large values of $S_n$ lead to the rejection of $H_0$, an approximate $p$-value for the test is finally given by $N^{-1} \sum_{k=1}^N\1(S_n^{(k)} \geq S_n)$.

When $n$, $p$ or $d$ are large, the above procedure can become very computationally expensive as, for every $k \in \{1,\dots,N\}$, it requires the generation of a random sample from $F_{\theta_n}$ and the estimation of $\theta$ from the generated data, both steps being potentially very time-consuming.

A computationally more efficient approach consists of using a {\em weighted bootstrap} in the sense of \cite{Bur00} and \cite{HorKokSte00} \citep[see also][and the references therein]{Hor00}. This resampling technique, based on the multiplier central limit theorem for empirical processes \citep[see e.g.][]{vanWel96,Kos08}, was recently used for assessing the goodness of fit of copula models in \cite{KojYanHol11}. While being asymptotically equivalent to the parametric bootstrap under the null hypothesis, it was found that, in the case of large samples, the use of the weighted instead of the parametric bootstrap could reduce the computing time from about a day to minutes for certain multivariate multiparameter models. 

The aim of this paper is to investigate, both theoretically and empirically, goodness-of-fit tests for multivariate distributions based on a weighted bootstrap in the sense of \cite{Bur00}. From a practical perspective, a large scale Monte Carlo study was carried out in order to approximately determine the sample size from which the parametric and the weighted bootstrap have roughly the same power. For small samples, the parametric bootstrap usually appears more powerful, and, since it typically has an acceptable computational cost in that case, it is the recommended approach. For larger samples, the use of the parametric bootstrap can become very tedious in practice and the derived faster {\em multiplier} goodness-of-fit procedure appears as a natural alternative. 

The paper is organized as follows. In the second section, to extend the breadth of the approach, the theoretical results establishing the validity of the approach are stated in the context of the theory of empirical processes as presented for instance in \cite{vanWel96}. The results of a large-scale simulation study are partially reported in the third section for univariate, bivariate and trivariate data sets and parametric c.d.f.\ families with up to nine parameters. The last section is devoted to a detailed illustration on real financial data. All the proofs and computational details are relegated to the appendices.

Note finally that the code of all the tests studied in this work will be documented and released as an R package accompanying the paper.

\section{The weighted bootstrap for goodness-of-fit testing}

To extend the breadth of our study, we shall work in the framework of the theory of empirical processes as presented for instance in \cite{vanWel96} or \cite{Kos08}. Given a random sample $X_1,\dots,X_n$ from a probability distribution $P$ on $\Rset^d$, the empirical measure is defined to be $\Pn = n^{-1} \sum_{i=1}^n \delta_{X_i}$, where $\delta_x$ is the measure that assigns a mass of 1 at $x$ and zero elsewhere. For a measurable function $f:\Rset^d \to \Rset$, $\Pn f$ denotes the expectation of $f$ under $\Pn$, and $P f$ the expectation under $P$, i.e.,
$$
\Pn f = \frac{1}{n} \sum_{i=1}^n f(X_i) \qquad \mbox{and} \qquad Pf = \int f \dd P.
$$
The empirical process evaluated at $f$ is then defined as $\Gn f = \sqrt{n} (\Pn f - Pf)$. 

As we continue, $\FF$ denotes a $P$-Donsker class of measurable functions, which means that the sequence of processes $\{ \Gn f : f \in \FF\}$ converges weakly to a $P$-Brownian bridge $\{ \G_P f : f \in \FF\}$ in the space  $\ell^\infty(\FF)$ of bounded functions from $\FF$ to $\Rset$ equipped with the uniform metric in the sense of Definition~1.3.3 of \citet{vanWel96}. Following usual notational conventions, this weak convergence will simply be denoted by $\Gn \leadsto \G_P$ in $\ell^\infty(\FF)$. By taking $\FF$ to be the class of indicator functions of lower-left orthants in $\Rset^d$, i.e., $\FF = \{x \mapsto \1(x \leq t) : t \in \overline{\Rset}^d\}$ with $\overline{\Rset} = \Rset \cup \{-\infty,\infty\}$, one recovers a more classical version of Donsker's theorem stating that $\sqrt{n} (F_n - F)$, where $F$ is the c.d.f.\ associated with $P$, converges weakly in the space $\ell^\infty(\overline{\Rset}^d)$ of bounded functions on $\overline{\Rset}^d$ to an $F$-Brownian bridge $\beta$, i.e., a tight centered Gaussian process with covariance function $E\{\beta(x)\beta(y)\} = F(x \wedge y) -  F(x)F(y)$, $x, y \in \Rset^d$.

The advantage of working in this general framework is that the forthcoming results remain valid for any $P$-Donsker class $\FF$. Although the collection of all indicator functions of lower-left orthants in $\Rset^d$ may appear as a natural choice for $\FF$ in the goodness-of-fit framework under consideration, other choices might be of interest such as the class of indicator functions of closed balls, rectangles or half-spaces \citep[see][for a related discussion regarding the choice of $\FF$]{Rom88}.

In the rest of the paper, convergence in probability is to be understood in outer probability \citep[see e.g.][Chapter 18]{van98}.

\subsection{Theoretical results}

Given i.i.d.\ random variables $Z_1,\dots,Z_n$ with mean 0, variance 1, satisfying $\int_0^\infty \{ P(|Z_1| > x) \}^{1/2} \dd x < \infty$, and independent of the random sample $X_1,\dots,X_n$, the following {\em multiplier} versions of $\Gn$ will be of interest:
\begin{equation}
\label{multiplier_versions}
\Gn' = \frac{1}{\sqrt{n}} \sum_{i=1}^n Z_i (\delta_{X_i} - P) \qquad \mbox{and} \qquad \Gn'' = \frac{1}{\sqrt{n}} \sum_{i=1}^n (Z_i - \bar Z) \delta_{X_i},
\end{equation}
where $\bar Z = n^{-1} \sum_{i=1}^n Z_i$. 

Let $\{P_\theta : \theta \in \OO\}$ be an identifiable family of distributions, where $\OO$ is an open subset of $\Rset^p$, and let $\{\psi_\theta : \theta \in \OO \}$ be a class of measurable functions from $\Rset^d$ to $\Rset^p$. Assume additionally that
\begin{enumerate}[({A}1)]
\item for any $\theta_0 \in \OO$, the map $\theta \mapsto P_\theta$ from $\Rset^p$ to $\ell^\infty(\FF)$ is Fr\'echet differentiable at $\theta_0$, i.e, there exists a map $\dot P_{\theta_0} : \FF \to \Rset^p$ such that
$$
 \sup_{f \in \FF} | P_{\theta}f - P_{\theta_0}f - (\theta - \theta_0)^\top \dot P_{\theta_0}f | = o( \| \theta - \theta_0 \| ) \qquad \mbox{as } \qquad \theta \to \theta_0,
$$
\item for any $\theta_0 \in \OO$, 
$$
\sup_{f \in \FF} \| \dot P_\theta f - \dot P_{\theta_0} f \| = o(1) \qquad \mbox{as } \qquad \theta \to \theta_0,
$$

\item for any $\theta_0 \in \OO$, there exists a $\delta > 0$ such that the class of measurable functions $\{\psi_\theta : \|\theta - \theta_0 \| < \delta \}$ is $P$-Donsker, 

\item for any $\theta_0 \in \OO$,
$$
\int \| \psi_{\theta}(x) - \psi_{\theta_0}(x) \|^2 \dd P(x) = o(1) \qquad \mbox{as } \qquad \theta \to \theta_0,
$$

\item for any $\theta_0 \in \OO$ and a random sample $X_1,\dots,X_n$ from $P_{\theta_0}$, $\theta_n$ is an estimator of $\theta_0$ that is asymptotically linear with influence function $\psi_{\theta_0}$, i.e.,
$$
\sqrt{n} (\theta_n - \theta_0) = \frac{1}{\sqrt{n}} \sum_{i=1}^n \psi_{\theta_0}(X_i) + o_{P_{\theta_0}}(1), 
$$ 
with $P_{\theta_0} \psi_{\theta_0} = 0$ and $P_{\theta_0} \| \psi_{\theta_0} \|^2 < \infty$.
\end{enumerate}

The following two propositions, proved in Appendix~\ref{proofs}, are at the root of the goodness-of-fit procedure to be given in the next subsection.

\begin{prop}
\label{propH0}
Let $X_1,\dots,X_n$ be a random sample from the distribution $P_{\theta_0}$ for some $\theta_0 \in \OO$. If Assumptions (A1)-(A5) are satisfied, then
$$
\left(f \mapsto  \sqrt{n} (\Pn f - P_{\theta_n}f ),f \mapsto  \Gn'f - \Gn' \psi_{\theta_n}^\top \dot P_{\theta_n}f, f \mapsto \Gn''f - \Gn'' \psi_{\theta_n}^\top \dot P_{\theta_n}f \right) 
$$
converges weakly to 
$$
\left(f \mapsto \G_{\theta_0}f - \G_{\theta_0} \psi_{\theta_0}^\top \dot P_{\theta_0}f , f \mapsto \G_{\theta_0}'f - \G_{\theta_0}' \psi_{\theta_0}^\top \dot P_{\theta_0}f, f \mapsto \G_{\theta_0}'f - \G_{\theta_0}' \psi_{\theta_0}^\top \dot P_{\theta_0}f \right)
$$
in $\{ \ell^\infty(\FF) \}^3$, where $\G_{\theta_0}$, a $P_{\theta_0}$-Brownian bridge, is the weak limit of $\Gn$, and $\G_{\theta_0}'$ is an independent copy of $\G_{\theta_0}$.
\end{prop}

\begin{prop}
\label{propH1}
Let $X_1,\dots,X_n$ be a random sample from a distribution $P \not \in \{P_\theta : \theta \in \OO\}$. If Assumptions (A1)-(A4) are satisfied, and if there exists $\theta_0 \in \OO$ such that $\sqrt{n} (\theta_n - \theta_0)$ converges in distribution under $P$, then
$$
\sup_{f \in \FF} | \sqrt{n} (\Pn f - P_{\theta_n} f) | \overset{P}{\to} \infty
$$
while
$$
 \left(f \mapsto \Gn'f - \Gn' \psi_{\theta_n}^\top \dot P_{\theta_n}f, f \mapsto \Gn''f - \Gn'' \psi_{\theta_n}^\top \dot P_{\theta_n}f \right) 
$$
converges weakly to 
$$
\left(f \mapsto \G_P'f - \G_P' \psi_{\theta_0}^\top \dot P_{\theta_0}f, f \mapsto \G_P'f - \G_P' \psi_{\theta_0}^\top \dot P_{\theta_0}f \right)
$$
in $\{ \ell^\infty(\FF) \}^2$, where $\G_P'$, a $P$-Brownian bridge, is the weak limit of $\Gn'$. 
\end{prop}

\subsection{The goodness-of-fit procedure}
\label{procedure}

Let us reformulate the results stated in Propositions~\ref{propH0} and~\ref{propH1} when $\FF$ is the class of indicator functions of lower-left orthants in $\Rset^d$. As the empirical process $\Gn'$ defined in~(\ref{multiplier_versions}) depends on the unknown true distribution $P$, we shall only deal with the parts of the above results concerned with $\Gn''$. Thus, let 
\begin{equation}
\label{multF}
\Fb_n''(x) = \frac{1}{\sqrt{n}} \sum_{i=1}^n (Z_i - \bar Z) \{ \1(X_i \leq x) -  \psi_{\theta_n}^\top(X_i) \dot F_{\theta_n}(x) \}, \qquad x \in \Rset^d,
\end{equation}
where $\{F_\theta : \theta \in \OO\} = \MM$ is the set of c.d.f.s associated with the parametric family of distributions $\{P_\theta : \theta \in \OO\}$, and where, for any $x \in \Rset^d$, $\dot F_\theta(x)$ is the gradient of $F_\theta(x)$ with respect to $\theta$. Proposition~\ref{propH0} then states that, if $P = P_{\theta_0}$ for some $\theta_0 \in \OO$, then, under explicit regularity conditions, $\Fb_n$ and $\Fb_n''$ jointly converge weakly in $\{\ell^\infty(\overline{\Rset}^d)\}^2$ to independent copies of the same limit. Roughly speaking, under the null hypothesis $H_0 : F \in \MM$, where $F$ is the c.d.f.\ associated to $P$, the empirical process $\Fb_n''$ is close to being an independent copy of $\Fb_n$. Every new set of multipliers $Z_1\dots,Z_n$ gives a new approximate independent copy of $\Fb_n$. 

Proposition~\ref{propH1} is concerned with the behavior of $\Fb_n$ and $\Fb_n''$ under the alternative hypothesis $H_1 : F \not \in \MM$. If $F \not \in \{F_{\theta} : \theta \in \OO\}$, then, under explicit regularity conditions, any sensible statistic derived from $\Fb_n$ will tend to infinity in probability because $ \sup_{x \in \Rset^d}| \Fb_n(x) | \overset{P}{\to} \infty$, while $\Fb_n''$ still converges weakly. 

The verification of Assumptions (A1)-(A5) for a given family $\MM$ of c.d.f.s is discussed in Section~\ref{checking}.

Now, for illustration purposes, let us assume that the test statistic is the Cram\'er--von Mises statistic $S_n$ defined in~(\ref{Sn}). The results of Propositions~\ref{propH0} and~\ref{propH1} reformulated above then suggest adopting the following goodness-of-fit procedure:
\begin{enumerate}

\item Estimate $\theta_0$ using an asymptotically linear estimator $\theta_n$.

\item Compute the test statistic $S_n = \int_{\Rset^d} \{ \Fb_n(x) \}^2 \dd F_{\theta_n}(x)$.

\item Then, for some large integer $N$, repeat the following steps for every $k \in \{1,\dots,N\}$:
  \begin{enumerate}
  \item Generate $n$ i.i.d.\ random variates $Z_1,\dots,Z_n$ with expectation 0, variance~1 and satisfying $\int_0^\infty \{ P(|Z_1| > x) \}^{1/2} \dd x < \infty$.
  \item Form an approximate realization of the test statistic under $H_0$ by
    $$
    S_n^{(k)} =  \int_{\Rset^d} \left\{ \Fb_n''(x) \right\}^2 \dd F_{\theta_n}(x), 
    $$
where $\Fb_n''$ is defined in~(\ref{multF})
  \end{enumerate}

\item An approximate $p$-value for the test is then given by $N^{-1} \sum_{k=1}^N \1(S_n^{(k)} \geq S_n)$. 

\end{enumerate}

The results given in the previous subsection imply that, under $H_0: F \in \MM$ and regularity conditions, the above testing procedure will hold its level asymptotically. Indeed, $(S_n,S_n^{(1)},\dots,S_n^{(N)})$ converge jointly in distribution to independent copies of the same limit and, thus, the approximate $p$-value computed in Step 4 of the procedure will be approximately standard uniform. On the other hand, under $H_1: F \not \in \MM$ and regularity conditions, $S_n$ tends in probability to infinity while $(S_n^{(1)},\dots,S_n^{(N)})$ still converges in distribution. It follows that the approximate $p$-value will tend to zero in probability.

The potential computational advantage of the multiplier goodness-of-fit procedure over the parametric bootstrap is best seen when Step 3 above is compared with the procedure recalled in the introduction. Roughly speaking, random number generation from the fitted distribution and estimation of the parameters from the generated data at each parametric bootstrap iteration is replaced by the random generation of $n$ multipliers, typically from the standard normal distribution or the uniform distribution on $\{-1,1\}$. 

To obtain even faster goodness-of-fits tests, we consider variations of $S_n$ and of the Kolmogorov-Smirnov statistic $T_n$ defined in~(\ref{Tn}) given by
\begin{equation}
\label{Sn'}
S_n^* =  \int_{\Rset^d}\{ \Fb_n(x) \}^2 \dd F_n(x) = \frac{1}{n} \sum_{i=1}^n \{ \Fb_n(X_i) \}^2 = \sum_{i=1}^n \{F_n(X_i) - F_{\theta_n}(X_i)\}^2,
\end{equation}
and
\begin{equation}
\label{Tn'}
T_n^* =  \max_{i \in \{1,\dots,n\}}  \left| \Fb_n(X_i) \right| = \sqrt{n} \max_{i \in \{1,\dots,n\}} \left|F_n(X_i) - F_{\theta_n}(X_i) \right|,
\end{equation}
respectively. When the goodness-of-fit procedure is based on $S_n^*$ (resp.\ $T_n^*$), the multiplier realizations of Step 3~(b) are thus simply given by
$$
S_n^{*,(k)} = \int_{\Rset^d}\{ \Fb_n''(x) \}^2 \dd F_n(x) = \frac{1}{n} \sum_{i=1}^n \{ \Fb_n''(X_i) \}^2 
\qquad \left(\mbox{resp. } T_n^{*,(k)} = \max_{i \in \{1,\dots,n\}} | \Fb_n''(X_i)  | \right).
$$

\subsection{About the assumptions}
\label{checking}

When $\FF$ is the class of indicator functions of lower-left orthants in $\Rset^d$, Assumptions (A1) and (A2) concern the existence and the smoothness of the gradient $\dot F_\theta$. As discussed for instance in \citet{GenHae06}, these two conditions are typically satisfied if, for any $\theta \in \OO$, $F_\theta$ has a p.d.f.\ $f_\theta$, and the function $(x,\theta) \mapsto f_\theta(x)$ is smooth in both $x$ and $\theta$. 

According to Example 19.7 in \cite{van98}, Assumption (A3) is satisfied if there exists a measurable function $m$ such that
\begin{equation}
\label{lipschitz}
\| \psi_{\theta_1}(x) - \psi_{\theta_2}(x) \| \leq m(x) \| \theta_1 - \theta_2 \| , \qquad \forall \, \theta_1, \theta_2 \in \{ \theta : \| \theta - \theta_0 \| < \delta \},
\end{equation}
with $P |m|^r < \infty$ for some $r > 0$. The above inequality is satisfied for many M-estimators with  $P m^2 < \infty$. Assume furthermore that the functions $\theta \mapsto \psi_\theta(x)$ are continuously differentiable and denote by $\dot \psi_\theta(x)$ the gradient of $\psi_\theta(x)$ with respect to $\theta$. Since $\theta \mapsto \dot \psi_\theta(x)$ is continuous, then, as explained in \citet[page 53]{van98}, the natural candidate for $m$ is $\dot \psi(x) = \sup_{\theta :  \| \theta - \theta_0 \| < \delta} \| \dot \psi_\theta(x) \|$. It then remains to verify that $P \dot \psi^2 < \infty$. In other words, Assumption (A3) is typically satisfied if there exists a square-integrable function  $\dot \psi$ such that $\| \dot \psi_\theta(x) \| \leq \dot \psi(x)$ when $\theta$ is close to $\theta_0$. Furthermore, if~(\ref{lipschitz}) holds, then, for $\| \theta - \theta_0 \| < \delta$,
$$
\int \| \psi_{\theta}(x) - \psi_{\theta_0}(x) \|^2 \dd P(x) \leq \| \theta - \theta_0 \| \int m^2(x) \dd P(x),
$$
and therefore, Assumption (A4) holds if $P m^2 < \infty$. 

Assumption~(A5) is related to the estimation of $\theta_0$ under the null hypothesis, and is typically established in the theorems proving the asymptotic normality of $M$-estimators. It follows that (A5) holds under the assumptions of these theorems \citep[see e.g.][Section 5.3]{van98}.

\section{Finite-sample performance}

In order to study the finite-sample performance of the proposed multiplier goodness-of-fit procedure, extensive Monte Carlo experiments were carried out in dimension one, two and three, for two- to nine-parameter families of absolutely continuous c.d.f.s. In all cases, the multipliers in the goodness-of-fit procedure given in Section~\ref{procedure} are taken from the standard normal distribution.

For a given hypothesized family $\MM = \{F_\theta : \theta \in \OO \}$, the estimation of the unknown parameter vector $\theta_0$ was performed using numerical maximum likelihood estimation based on the Nelder-Mead optimization algorithm as implemented in the R {\tt optim} routine \citep{Rsystem}. Method-of-moment estimates of $\theta_0$ were used as starting values. 

The use of maximum likelihood estimation implies, under classical regularity conditions \citep[see e.g.][Section 5.5]{van98}, that, for any $\theta_0 \in \OO$, the influence function $\psi_{\theta_0}$ appearing in Assumption~(A5) is given by $$
\psi_{\theta_0}(x) = I_{\theta_0}^{-1} \frac{\dot f_{\theta_0}(x)}{f_{\theta_0}(x)} \1 \{f_{\theta_0}(x) > 0 \},
$$ 
where $f_{\theta_0}$ is the probability density function (p.d.f.) associated with $F_{\theta_0}$, $\dot f_{\theta_0}(x)$ is the gradient of $f_{\theta_0}(x)$ with respect to~$\theta_0$, and $I_{\theta_0}$ is the Fisher information matrix. Equivalently, $\psi_{\theta_0}(x)$ can be taken as $I_{\theta_0}^{-1}$ times the gradient of $\log f_{\theta_0}(x)$ with respect to~$\theta_0$. 

From Section~\ref{procedure}, and in particular~(\ref{multF}), we see that $\psi_{\theta_n}$ is to be computed instead of the unknown function $\psi_{\theta_0}$. To obtain a more generic implementation, all gradients, including $\dot F_{\theta_n}$, were computed numerically using Richardson's extrapolation method as implemented in the R package {\tt numDeriv} \citep{numDeriv}. As shall be explained in more detail in the forthcoming subsections, this numerical approach was compared with a more precise implementation when $\MM$ is the set of c.d.f.s of the multivariate $t$ distribution with fixed degrees of freedom (d.f.). In both cases, the matrix $I_{\theta_n}$ was estimated as the sample covariance matrix of the sample $\dot f_{\theta_n}(X_1)/f_{\theta_n}(X_1),\dots,\dot f_{\theta_n}(X_n)/f_{\theta_n}(X_n)$. 

\subsection{Univariate experiments}

In dimension one, four candidate test statistics were considered: numerical approximations of $S_n$ and $T_n$ defined in~(\ref{Sn}) and~(\ref{Tn}), respectively, and the statistics $S_n^*$ and $T_n^*$ defined in~(\ref{Sn'}) and~(\ref{Tn'}), respectively. 

The numerical approximations of $S_n$ and $T_n$ are based on a uniform grid $u_1,\dots,u_m$ on $(0,1)$ which is transformed as $y_i = F^{-1}_{\theta_n}(u_i)$, $i \in \{1,\dots,m\}$. The goodness-of-fit procedure given in Section~\ref{procedure} is then based on the comparison of
$$
S_n \approx \frac{1}{m} \sum_{i=1}^m \{ \Fb_n(y_i) \}^2 
\qquad \left(\mbox{resp. } T_n \approx \max_{i \in \{1,\dots,m\}} | \Fb_n(y_i)  | \right),
$$
with the set of $N$ multiplier realizations
$$
S_n^{(k)} \approx \frac{1}{m} \sum_{i=1}^m \{ \Fb_n''(y_i) \}^2 \qquad \left(\mbox{resp. } T_n^{(k)} \approx \max_{i \in \{1,\dots,m\}} | \Fb_n''(y_i)  | \right), \qquad k \in \{1,\dots,N\}, 
$$
where $\Fb_n''$ is defined in~(\ref{multF}). The grid size $m$ was set to 1000 in our experiments. 

In terms of distributions, five two-parameter families of absolutely continuous c.d.f.s were considered, namely the normal, $t$, logistic, gamma and Weibull distributions. The d.f.\ of the $t$ distribution were fixed (to five) to avoid numerical issues that arise when attempting to estimate them \citep[see e.g.][for a review]{NadKot08}. These distributions are abbreviated by N, T5, L, G and W as we continue.

For data generation, the expectation and the variance of N were set to 10 and 1, respectively. For each of the remaining four distributions, its parameters were determined as approximate minimizers of the Kullback-Leibler divergence between the p.d.f.\ of the distribution and the p.d.f.\ of the normal $N(10,1)$. The expectation of T5  was thus fixed to 10 and its dispersion (scale) parameter to 0.856. The location and scale parameters of L were fixed to 10 and 0.572, respectively. The shape and rate parameters of G were set to 98.671 and 9.866, respectively, while, for W, the shape and scale parameters were fixed to 10.618 and 10.452, respectively. The parametrization of L, G and W used in this work corresponds to the implementation of these distributions in the R statistical system. A plot of the p.d.f.s of the five data generating distributions is given in Figure~\ref{uniden}.

\begin{figure}[t!]
  \begin{center}
    \includegraphics*[width=0.5\linewidth]{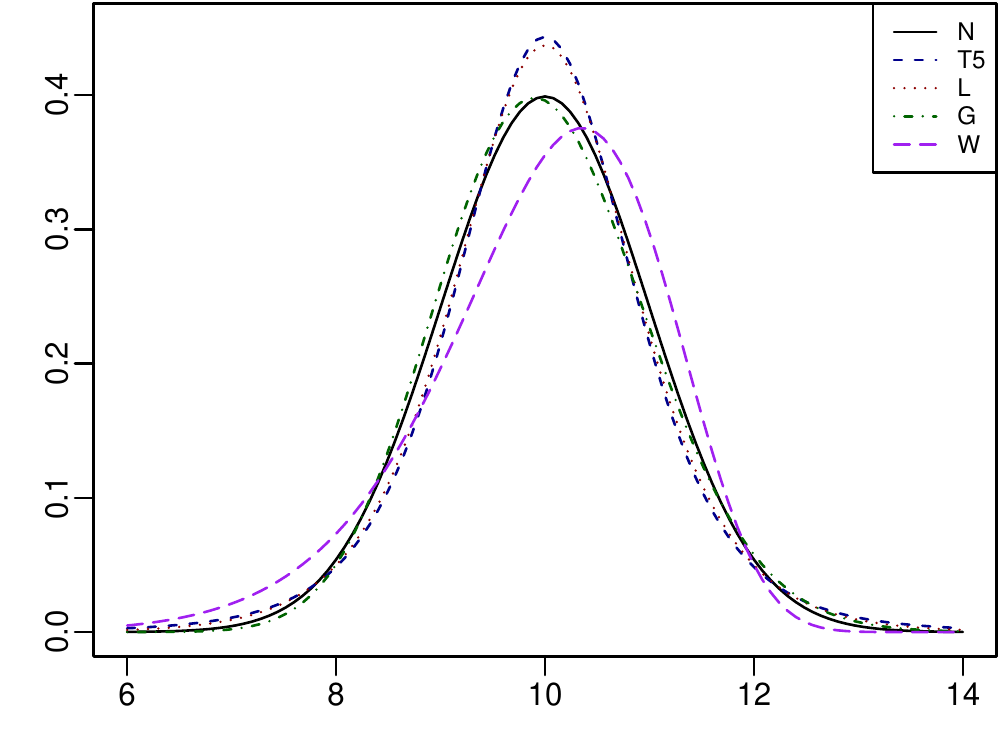}
  \end{center}
  \caption{Plot of the five p.d.f.s used for data generation in the first univariate experiment. The normal distribution is the $N(10,1)$, while the parameters of the remaining four distributions were determined as approximate minimizers of the Kullback-Leibler divergence between the distribution and the normal.}
  \label{uniden}
\end{figure}

For each data generating distribution, 1000 random samples of size $n$ were generated for $n=100$, 200, 300, 400 and 500. Under each scenario, the five families mentioned above were hypothesized. For each of the four test statistics, the proposed multiplier goodness-of-fit procedure (abbreviated by MP as we continue) was compared with the parametric bootstrap-based goodness-of-fit procedure (abbreviated by PB). All the tests were carried out at the 5\% level of significance. As the Cram\'er--von Mises statistics $S_n$ and $S_n^*$ consistently outerperformed the Kolmogorov-Smirnov statistics $T_n$ and $T_n^*$, we only report the rejection rates of the former in Table~\ref{dim1}. The five horizontal blocks of the table correspond to the five data generating distributions whose parameters were given previously. The first two columns contain the empirical powers of the Shapiro-Wilk and Jarque-Bera tests of univariate normality. The empirical levels of all the tests are in italic in the table. As can be noticed, these are, overall, reasonably close to the 5\% significance level even for small~$n$.

\begin{sidewaystable}[tbp]
\addtolength{\tabcolsep}{-3pt}
\begin{center}
\caption{Rejection rate (in \%) of the null hypothesis in the univariate case as observed in 1000 random samples of size $n=100$, 200, 300, 400 and 500.}
\label{dim1}
\begin{tabular}{cc rrrrrr rrrr rrrr rrrr rrrr}
\toprule
True &  & \multicolumn{6}{c}{N} & \multicolumn{4}{c}{T5} & \multicolumn{4}{c}{L} & \multicolumn{4}{c}{G} & \multicolumn{4}{c}{W}\\
\cmidrule(lr){3-8}\cmidrule(lr){9-12}\cmidrule(lr){13-16}\cmidrule(lr){17-20}\cmidrule(lr){21-24}
dist  & $n$  &    &    & \multicolumn{2}{c}{$S_n$} & \multicolumn{2}{c}{$S_n^*$} & \multicolumn{2}{c}{$S_n$} & \multicolumn{2}{c}{$S_n^*$} & \multicolumn{2}{c}{$S_n$} & \multicolumn{2}{c}{$S_n^*$} & \multicolumn{2}{c}{$S_n$} & \multicolumn{2}{c}{$S_n^*$} & \multicolumn{2}{c}{$S_n$} & \multicolumn{2}{c}{$S_n^*$} \\
\cmidrule(lr){5-6}\cmidrule(lr){7-8}\cmidrule(lr){9-10}\cmidrule(lr){11-12}\cmidrule(lr){13-14}\cmidrule(lr){15-16}\cmidrule(lr){17-18}\cmidrule(lr){19-20}\cmidrule(lr){21-22}\cmidrule(lr){23-24}
  &   & SH & JB & MP & PB & MP & PB & MP & PB & MP & PB & MP & PB & MP & PB & MP & PB & MP & PB & MP & PB & MP & PB\\
  \midrule
N & 100 & {\it   4.0} & {\it   3.8} & {\it   4.7} & {\it   3.7} & {\it   4.8} & {\it   4.4} &  12.5 &  12.4 &  11.7 &  12.5 &  10.2 &  10.0 &   9.9 &   9.8 &   9.0 &   8.2 &   6.3 &   6.4 &  44.7 &  49.2 &  45.6 &  53.7 \\ 
   & 200 & {\it   4.6} & {\it   4.9} & {\it   4.9} & {\it   5.3} & {\it   5.3} & {\it   5.0} &  27.2 &  23.3 &  26.7 &  23.4 &  21.5 &  17.4 &  22.0 &  17.2 &  11.9 &  11.4 &   9.3 &   9.1 &  78.8 &  83.5 &  80.6 &  85.6 \\ 
   & 300 & {\it   5.8} & {\it   5.2} & {\it   5.6} & {\it   5.4} & {\it   4.8} & {\it   5.5} &  36.4 &  36.3 &  36.4 &  36.7 &  25.8 &  25.8 &  26.6 &  25.7 &  16.0 &  17.1 &  13.1 &  14.2 &  93.2 &  94.3 &  93.8 &  95.5 \\ 
   & 400 & {\it   5.1} & {\it   4.9} & {\it   5.4} & {\it   4.8} & {\it   5.3} & {\it   5.3} &  51.5 &  47.8 &  51.7 &  47.6 &  38.6 &  34.8 &  39.7 &  34.3 &  19.5 &  22.8 &  17.7 &  19.0 &  98.5 &  98.9 &  98.9 &  99.1 \\ 
   & 500 & {\it   4.6} & {\it   4.3} & {\it   5.5} & {\it   5.9} & {\it   5.6} & {\it   6.2} &  64.2 &  58.3 &  64.6 &  58.4 &  47.5 &  40.6 &  48.0 &  40.7 &  24.5 &  24.4 &  21.1 &  22.1 &  99.6 &  99.5 &  99.6 &  99.5 \\ 
   [1ex]
T5 & 100 &  54.5 &  61.7 &  30.7 &  39.1 &  25.1 &  38.7 & {\it   4.1} & {\it   5.5} & {\it   3.2} & {\it   5.3} &   5.2 &   7.5 &   3.3 &   7.5 &  38.3 &  46.7 &  30.6 &  43.5 &  75.7 &  79.4 &  73.7 &  80.4 \\ 
   & 200 &  81.6 &  86.0 &  59.5 &  67.3 &  55.0 &  66.8 & {\it   4.7} & {\it   6.6} & {\it   4.3} & {\it   6.2} &   8.5 &   9.7 &   7.8 &   9.5 &  67.4 &  74.4 &  63.0 &  72.7 &  97.9 &  98.5 &  97.9 &  98.6 \\ 
   & 300 &  93.0 &  95.6 &  78.9 &  80.3 &  75.4 &  79.7 & {\it   4.0} & {\it   5.1} & {\it   3.7} & {\it   5.4} &   7.9 &  11.5 &   7.6 &  11.1 &  85.7 &  85.7 &  82.8 &  84.7 &  99.8 &  99.9 &  99.8 &  99.9 \\ 
   & 400 &  97.4 &  98.6 &  88.6 &  93.0 &  86.2 &  93.2 & {\it   4.8} & {\it   5.1} & {\it   4.5} & {\it   5.0} &  10.0 &  14.8 &   9.4 &  14.9 &  92.4 &  95.9 &  90.8 &  95.6 & 100.0 & 100.0 & 100.0 & 100.0 \\ 
   & 500 &  99.0 &  99.4 &  94.6 &  96.4 &  93.3 &  96.6 & {\it   4.5} & {\it   5.0} & {\it   4.5} & {\it   4.8} &  12.0 &  14.6 &  11.9 &  14.7 &  97.6 &  97.4 &  96.8 &  97.4 & 100.0 & 100.0 & 100.0 & 100.0 \\ 
   [1ex]
L & 100 &  30.7 &  37.2 &  14.4 &  21.1 &  10.6 &  19.7 &   4.7 &   6.4 &   4.1 &   6.3 & {\it   4.5} & {\it   6.4} & {\it   4.1} & {\it   6.3} &  21.2 &  27.9 &  14.9 &  24.5 &  66.0 &  73.4 &  64.6 &  75.6 \\ 
   & 200 &  48.3 &  56.5 &  26.5 &  33.2 &  22.0 &  32.5 &   4.7 &   6.2 &   4.5 &   6.2 & {\it   4.0} & {\it   5.9} & {\it   3.7} & {\it   5.8} &  37.8 &  47.8 &  31.9 &  45.0 &  94.3 &  95.6 &  94.3 &  96.2 \\ 
   & 300 &  63.9 &  71.5 &  38.2 &  48.6 &  35.1 &  47.6 &   6.6 &   5.1 &   6.6 &   5.1 & {\it   4.7} & {\it   4.7} & {\it   4.6} & {\it   4.7} &  53.3 &  63.1 &  47.5 &  61.4 &  99.2 &  99.2 &  99.2 &  99.4 \\ 
   & 400 &  75.0 &  80.9 &  53.1 &  59.9 &  49.9 &  59.3 &   7.4 &   7.0 &   7.1 &   7.2 & {\it   4.7} & {\it   5.0} & {\it   4.6} & {\it   5.1} &  70.0 &  74.5 &  65.2 &  73.2 &  99.7 & 100.0 &  99.8 & 100.0 \\ 
   & 500 &  82.7 &  88.1 &  63.5 &  67.6 &  59.7 &  67.2 &   8.2 &   9.4 &   8.1 &   9.4 & {\it   4.8} & {\it   6.0} & {\it   4.8} & {\it   6.1} &  80.3 &  81.8 &  77.2 &  79.9 & 100.0 & 100.0 &  99.9 & 100.0 \\ 
   [1ex]
G & 100 &  12.1 &   9.3 &   8.6 &   9.1 &  10.8 &  10.7 &  13.7 &  14.6 &  12.5 &  13.7 &  11.6 &  11.8 &  11.7 &  11.8 & {\it   5.1} & {\it   5.5} & {\it   5.5} & {\it   5.8} &  67.9 &  75.2 &  69.3 &  78.1 \\ 
   & 200 &  17.0 &  13.9 &  11.7 &  13.2 &  13.8 &  14.8 &  29.9 &  24.7 &  28.1 &  24.2 &  23.2 &  19.7 &  23.5 &  19.7 & {\it   5.6} & {\it   5.5} & {\it   5.2} & {\it   5.7} &  95.1 &  96.3 &  95.6 &  97.2 \\ 
   & 300 &  23.6 &  22.1 &  14.8 &  15.0 &  17.8 &  17.3 &  43.2 &  40.0 &  41.4 &  39.0 &  32.9 &  31.0 &  33.0 &  30.9 & {\it   4.9} & {\it   5.3} & {\it   5.5} & {\it   5.4} &  99.7 &  99.4 &  99.8 &  99.4 \\ 
   & 400 &  29.4 &  27.0 &  17.8 &  20.5 &  20.0 &  21.8 &  57.8 &  53.4 &  57.1 &  52.5 &  45.1 &  41.0 &  46.1 &  41.0 & {\it   5.7} & {\it   3.5} & {\it   5.4} & {\it   3.6} & 100.0 & 100.0 & 100.0 & 100.0 \\ 
   & 500 &  35.8 &  32.5 &  22.9 &  24.2 &  25.1 &  26.6 &  68.1 &  66.6 &  66.9 &  65.8 &  55.2 &  51.0 &  55.9 &  51.2 & {\it   4.8} & {\it   3.7} & {\it   4.8} & {\it   4.0} & 100.0 & 100.0 & 100.0 & 100.0 \\ 
   [1ex]
W & 100 &  65.9 &  53.0 &  44.5 &  52.6 &  34.1 &  42.9 &  24.1 &  34.0 &  23.2 &  37.4 &  22.3 &  33.0 &  20.0 &  32.2 &  70.1 &  76.7 &  59.4 &  70.3 & {\it   4.3} & {\it   4.3} & {\it   4.4} & {\it   4.1} \\ 
   & 200 &  94.1 &  88.9 &  81.1 &  83.5 &  74.6 &  79.2 &  62.4 &  66.2 &  63.5 &  69.0 &  59.8 &  63.7 &  58.5 &  63.3 &  96.3 &  97.2 &  94.9 &  96.0 & {\it   5.2} & {\it   5.5} & {\it   5.2} & {\it   5.2} \\ 
   & 300 &  98.8 &  98.1 &  94.5 &  93.7 &  92.8 &  91.4 &  82.3 &  84.7 &  84.0 &  85.6 &  78.5 &  80.6 &  78.6 &  80.6 &  99.4 & 100.0 &  99.2 &  99.9 & {\it   6.0} & {\it   4.6} & {\it   5.5} & {\it   4.8} \\ 
   & 400 &  99.9 &  99.7 &  98.5 &  98.6 &  97.7 &  97.8 &  95.1 &  94.9 &  95.8 &  96.0 &  92.7 &  92.6 &  92.3 &  92.6 &  99.9 &  99.9 &  99.9 &  99.9 & {\it   5.2} & {\it   4.4} & {\it   5.2} & {\it   3.8} \\ 
   & 500 & 100.0 & 100.0 &  99.0 & 100.0 &  98.9 & 100.0 &  98.3 &  99.0 &  98.5 &  99.1 &  97.3 &  98.1 &  97.2 &  98.1 & 100.0 & 100.0 & 100.0 & 100.0 & {\it   4.6} & {\it   4.4} & {\it   4.5} & {\it   4.3} \\ 
\bottomrule
\end{tabular}
\end{center}
\end{sidewaystable}  

In terms of rejection rates, the tests based on $S_n$ and $S_n^*$ are, overall, very close. From a practical perspective, recall that the latter are easier to implement. Notice also that the tests based on PB are, overall, more powerful than those based on MP for small $n$. However, as $n$ increases, the difference in rejection rate vanishes in most scenarios. These results suggest that the multiplier procedure can be safely used as a large-sample alternative to the parametric bootstrap in dimension one.

To complement the previous study, the levels of the tests were also investigated for heavy-tailed and strongly asymmetric distributions. More specifically, the levels of the tests were estimated from 1000 random samples of size $n$ generated from the standard $t$ distribution with fixed d.f.\ $\nu \in \{1, 2, 3, 4, 5\}$, and from the gamma distribution with rate parameter 0.5 and shape parameter in $\{1, 2, 4, 8, 16\}$. The obtained rejection rates (not reported) were found to be reasonably close to the 5\% nominal level in all scenarios and for all $n \in \{100,200,300,400,500\}$.

\subsection{Bivariate and trivariate experiments}

In the bivariate and trivariate simulations, only the goodness-of-fit procedures based on $S_n^*$ and $T_n^*$ were used, and five families of distributions were considered. In addition to the multivariate normal (abbreviated by N) and the multivariate $t$ distribution with five d.f.\ (abbreviated by T5), three absolutely continuous families were constructed from \citet{Skl59}'s representation theorem. The latter result states that any multivariate c.d.f.\ $F:\Rset^d \to [0,1]$ whose marginal c.d.f.s $F_1,\dots,F_d$ are continuous can be expressed in terms of a unique $d$-dimensional copula $C$ as
$$
F(x) = C\{F_1(x_1),\dots,F_d(x_d)\}, \qquad x = (x_1,\dots,x_d) \in \Rset^d.
$$
A first non-Gaussian family was obtained by taking $F_1,\dots,F_d$ to be gamma and $C$ to be a normal copula, a second choice was to take $F_1,\dots,F_d$ to be $t$ with five d.f.\ and $C$ to be a normal copula, while a third family was obtained by taking $F_1,\dots,F_d$ normal and $C$ to be a Clayton copula.  These three families are respectively abbreviated by GN, T5N and NC as we continue. Notice that the multivariate normal N is obtained by taking $F_1,\dots,F_d$ normal and $C$ normal, while T5, the multivariate $t$ with five d.f., is obtained by taking $F_1,\dots,F_d$ to be $t$ with five d.f.\ and $C$ to be a $t$ copula with five d.f.

In dimension two, all five families have five parameters: two parameters per margin and one parameter for the copula. In dimension three, N, T5, GN and T5N have nine parameters (two parameters per margin and three correlation coefficients for the copulas), while NC has seven parameters (two parameters per margin and one parameter for the Clayton copula). 

For data generation, the margins of the distributions N and NC were taken to be the $N(10,1)$, the gamma margins of GN were chosen to have shape and rate parameters equal to 98.671 and 9.866, respectively, while the margins of T5 and T5N were set to have expectation 10 and dispersion 0.856. In order to study the effect of the dependence on the tests, the correlation coefficients of the normal and $t$ copulas were first taken equal to 0.309 and then to 0.588, while the parameter of the Clayton copula used to construct NC was first taken equal to 0.5 and then to 1.333. This corresponds to requiring that the value of Kendall's tau (denoted by $\tau$ in the tables) for all bivariate margins of the copulas is first equal to 0.2 and then to 0.4.

For N and T5, random number generation and the computation of the p.d.f.\ and the c.d.f.\ was performed using the excellent {\tt mvtnorm} R package \citep{mvtnorm}. For NC, GN and T5N, the {\tt copula} package \citep{KojYan10} was used in addition.

The results for dimension two are given in Table~\ref{dim2}. The columns SW contain the rejection rates of the multivariate extension of the Shapiro--Wilk test proposed by \citet{VilGon09} and implemented in the R package {\tt mvShap\-iroTest} \citep{mvShapiroTest}. Unlike in dimension one, the comparison between the multiplier procedure and the parametric bootstrap was carried out only in the situation where a bivariate normal distribution is hypothesized. In that case, the maximum likelihood estimates are the sample mean and the sample covariance matrix, and these were used in the parametric bootstrap to decrease its computational cost. The rejection rates of the multiplier procedure and the parametric bootstrap are reported in the columns N-MP and N-PB, respectively. The parametric bootstrap appears clearly more powerful than the multiplier but, as in dimension one, the difference in rejection rate tends to vanish as $n$ reaches 500. Also, the comparison of the columns SW, N-MP and N-PB reveals, without surprise, that the specialized multivariate Shapiro--Wilk test is generally more powerful than its two generic competitors.

\begin{table}[t!]
\addtolength{\tabcolsep}{-3pt}
\caption{Rejection rate (in \%) of the null hypothesis in the bivariate case as observed in 1000 random samples of size $n=100$, 200, 300, 400 and 500.}
\label{dim2}
\begin{center}
\begin{tabular}{lr rrrrrrr rrrrrrr}
  \toprule
True &  & \multicolumn{7}{c}{$\tau = 0.2$} & \multicolumn{7}{c}{$\tau = 0.4$}\\
\cmidrule(lr){3-9}\cmidrule(lr){10-16}
dist  & $n$  &  SW & N-PB & N-MP & NC & T5N & T5 & GN & SW & N-PB & N-MP & NC & T5N & T5 & GN \\
  \midrule
N & 100 & {\it   3.4} & {\it   4.6} & {\it   4.3} &  15.3 &   7.7 &   8.8 &   3.6 & {\it   4.8} & {\it   4.3} & {\it   3.4} &  49.6 &   7.2 &   8.0 &   5.2 \\ 
   & 200 & {\it   5.9} & {\it   4.4} & {\it   4.1} &  30.9 &  21.4 &  22.5 &   8.1 & {\it   4.9} & {\it   4.2} & {\it   4.0} &  85.1 &  15.0 &  20.0 &   8.7 \\ 
   & 300 & {\it   5.6} & {\it   4.5} & {\it   4.0} &  42.8 &  36.8 &  38.7 &  11.7 & {\it   5.4} & {\it   4.5} & {\it   4.4} &  96.0 &  26.6 &  33.9 &  14.5 \\ 
   & 400 & {\it   5.5} & {\it   4.3} & {\it   4.1} &  55.6 &  47.9 &  49.6 &  18.2 & {\it   4.3} & {\it   5.7} & {\it   6.0} &  99.5 &  36.8 &  48.5 &  16.7 \\ 
   & 500 & {\it   5.3} & {\it   5.6} & {\it   5.3} &  62.7 &  61.7 &  64.4 &  20.9 & {\it   4.8} & {\it   4.5} & {\it   4.7} & 100.0 &  45.4 &  58.5 &  21.6 \\ 
   [1ex]
NC & 100 &   6.0 &  11.3 &   7.0 & {\it   3.7} &  20.2 &  23.1 &   8.2 &  31.0 &  26.6 &  12.3 & {\it   4.0} &  51.3 &  35.1 &  11.0 \\ 
   & 200 &   6.3 &  16.5 &  11.1 & {\it   4.0} &  60.5 &  59.5 &  20.9 &  58.0 &  43.7 &  24.6 & {\it   6.0} &  95.4 &  84.7 &  33.7 \\ 
   & 300 &  10.9 &  22.9 &  15.7 & {\it   4.8} &  82.1 &  80.0 &  40.8 &  77.2 &  62.5 &  40.0 & {\it   5.6} &  99.9 &  98.8 &  56.4 \\ 
   & 400 &  12.8 &  25.6 &  17.7 & {\it   4.3} &  95.6 &  93.4 &  51.4 &  91.5 &  75.2 &  56.2 & {\it   6.1} & 100.0 & 100.0 &  75.8 \\ 
   & 500 &  13.7 &  36.5 &  26.8 & {\it   4.7} &  99.0 &  98.2 &  64.3 &  96.2 &  85.8 &  69.0 & {\it   6.2} & 100.0 & 100.0 &  85.7 \\ 
   [1ex]
T5N & 100 &  76.5 &  50.6 &  30.5 &  53.8 & {\it   3.0} &   2.7 &  26.0 &  78.0 &  51.8 &  29.8 &  84.5 & {\it   2.1} &   2.6 &  26.0 \\ 
   & 200 &  96.2 &  78.6 &  65.3 &  90.9 & {\it   2.9} &   4.4 &  71.3 &  95.7 &  79.1 &  64.4 &  99.6 & {\it   4.3} &   5.6 &  67.7 \\ 
   & 300 &  99.5 &  92.3 &  85.9 &  98.9 & {\it   4.2} &   6.5 &  90.3 &  99.7 &  92.0 &  83.5 & 100.0 & {\it   4.7} &   8.9 &  87.3 \\ 
   & 400 &  99.8 &  97.0 &  94.6 & 100.0 & {\it   3.7} &   7.7 &  97.0 &  99.9 &  97.3 &  93.9 & 100.0 & {\it   3.8} &   9.4 &  96.7 \\ 
   & 500 & 100.0 &  99.3 &  97.7 &  99.9 & {\it   4.4} &   9.4 &  99.4 & 100.0 &  99.2 &  98.3 & 100.0 & {\it   4.7} &  13.5 &  98.8 \\ 
   [1ex]
T5 & 100 &  73.5 &  55.4 &  35.7 &  60.0 &   3.9 & {\it   2.2} &  38.9 &  75.1 &  56.3 &  33.8 &  83.6 &   4.5 & {\it   2.9} &  33.7 \\ 
   & 200 &  93.8 &  83.7 &  69.7 &  91.9 &   4.7 & {\it   3.8} &  76.5 &  93.9 &  80.8 &  67.4 &  99.3 &   4.2 & {\it   2.6} &  72.5 \\ 
   & 300 &  99.1 &  95.7 &  90.1 &  98.8 &   5.1 & {\it   2.9} &  93.8 &  99.2 &  93.8 &  86.8 & 100.0 &   6.8 & {\it   4.2} &  91.2 \\ 
   & 400 & 100.0 &  98.8 &  96.1 &  99.8 &   5.5 & {\it   3.6} &  98.4 &  99.6 &  97.8 &  94.3 & 100.0 &   5.9 & {\it   4.2} &  98.0 \\ 
   & 500 & 100.0 & 100.0 &  99.8 & 100.0 &   5.7 & {\it   3.5} & 100.0 & 100.0 &  99.5 &  98.4 & 100.0 &   5.4 & {\it   3.1} &  99.6 \\ 
   [1ex]
GN & 100 &  14.3 &  13.7 &  11.1 &  23.9 &   7.8 &   8.2 & {\it   4.5} &  15.8 &  14.4 &   9.8 &  60.5 &   5.6 &   7.0 & {\it   3.9} \\ 
   & 200 &  21.8 &  19.8 &  18.3 &  49.7 &  20.9 &  22.1 & {\it   5.2} &  22.0 &  17.8 &  16.1 &  90.1 &  17.9 &  21.7 & {\it   4.6} \\ 
   & 300 &  33.5 &  24.3 &  22.7 &  68.6 &  39.7 &  40.3 & {\it   4.6} &  34.1 &  23.7 &  22.6 &  98.1 &  31.5 &  39.1 & {\it   4.6} \\ 
   & 400 &  42.8 &  29.9 &  29.0 &  81.6 &  56.3 &  60.4 & {\it   4.5} &  44.6 &  29.4 &  27.1 &  99.8 &  49.0 &  56.0 & {\it   5.0} \\ 
   & 500 &  57.3 &  38.1 &  37.1 &  90.3 &  70.7 &  74.5 & {\it   4.7} &  52.9 &  35.0 &  34.0 & 100.0 &  62.2 &  67.4 & {\it   4.8} \\ 
  \bottomrule
\end{tabular}
\end{center}
\end{table}

By looking at the entries in italic in Table~\ref{dim2}, we see that the multiplier procedure is, overall, too conservative for smaller $n$, but that the agreement between the empirical levels and the 5\% significance level improves as $n$ increases. The effect of stronger dependence on the power is variable. For instance, when data are generated from NC, the families N, T5N, T5 and GMN are easier to reject for $\tau=0.4$ than for $\tau=0.2$, but, when the true distribution is N, it is easier to reject T5N and T5 in the case of weaker dependence. Notice finally that, as could have been expected, it is very difficult to distinguish between T5 from T5N for these sample sizes.

The results of the trivariate Monte Carlo simulations are given in Table~\ref{dim3}. This time, to make the computational cost of the simulation acceptable, only the multiplier procedure was used. By looking at the entries in italic, we can see, as in the bivariate case, that the tests are, overall, too conservative, but that the empirical levels improve as $n$ increases. A comparison with Table~\ref{dim2} also reveals that the rejection rates are higher in dimension three than in dimension two, which suggests that the differences between the distributions are easier to detect as $d$ increases from 2 to 3. From a more practical perspective, we see that the empirical powers approach 100\% as $n$ reaches 500 in many scenarios under the alternative hypothesis.

\begin{table}[ht]
\caption{Rejection rate (in \%) of the null hypothesis in the trivariate case as observed in 1000 random samples of size $n=100$, 200, 300, 400 and 500.}
\label{dim3}
\begin{center}
\begin{tabular}{lr rrrrrr rrrrrr}
  \toprule
True & & \multicolumn{6}{c}{$\tau = 0.2$} & \multicolumn{6}{c}{$\tau = 0.4$}\\
\cmidrule(lr){3-8}\cmidrule(lr){9-14}
dist  & $n$ &  SW & N & NC & T5N & T5 & GN & SW & N & NC & T5N & T5 & GN \\
  \midrule
N & 100 & {\it   5.0} & {\it   2.2} &  42.7 &   4.8 &   5.9 &   1.6 & {\it   5.2} & {\it   2.1} &  93.7 &   4.4 &   5.8 &   2.0 \\ 
   & 200 & {\it   5.1} & {\it   3.6} &  75.8 &  18.2 &  20.8 &   4.6 & {\it   5.6} & {\it   3.2} &  99.8 &  12.9 &  19.3 &   5.3 \\ 
   & 300 & {\it   6.0} & {\it   3.9} &  90.3 &  34.0 &  35.7 &   8.8 & {\it   7.1} & {\it   4.2} &  99.9 &  24.4 &  31.9 &  11.2 \\ 
   & 400 & {\it   4.7} & {\it   3.7} &  97.3 &  46.1 &  48.6 &  16.8 & {\it   5.6} & {\it   5.6} & 100.0 &  33.2 &  42.9 &  16.8 \\ 
   & 500 & {\it   5.2} & {\it   3.8} &  99.6 &  57.6 &  62.0 &  20.6 & {\it   5.7} & {\it   3.8} & 100.0 &  45.7 &  56.7 &  20.6 \\ 
   [1ex]
NC & 100 &  11.1 &   5.5 & {\it   3.9} &  35.3 &  32.7 &   8.9 &  80.0 &  16.5 & {\it   3.4} &  93.5 &  62.0 &  21.1 \\ 
   & 200 &  17.6 &  22.5 & {\it   3.6} &  90.9 &  83.9 &  44.6 &  98.9 &  54.7 & {\it   5.2} & 100.0 &  99.2 &  66.3 \\ 
   & 300 &  27.4 &  41.4 & {\it   5.1} &  99.2 &  96.6 &  71.1 & 100.0 &  81.7 & {\it   3.8} & 100.0 & 100.0 &  92.1 \\ 
   & 400 &  37.9 &  53.6 & {\it   4.6} &  99.8 &  99.7 &  85.9 & 100.0 &  95.2 & {\it   5.9} & 100.0 & 100.0 &  97.1 \\ 
   & 500 &  49.6 &  69.0 & {\it   4.5} & 100.0 &  99.9 &  92.7 & 100.0 &  98.5 & {\it   5.1} & 100.0 & 100.0 &  99.8 \\ 
   [1ex]
T5N & 100 &  87.2 &  23.7 &  82.5 & {\it   1.6} &   3.2 &  15.0 &  86.8 &  21.6 &  98.9 & {\it   2.7} &   5.1 &  12.5 \\ 
   & 200 &  98.6 &  68.1 &  99.4 & {\it   2.7} &   9.3 &  66.6 &  98.9 &  67.9 & 100.0 & {\it   3.4} &  13.4 &  61.3 \\ 
   & 300 &  99.9 &  87.6 &  99.9 & {\it   2.9} &  15.0 &  88.4 & 100.0 &  85.9 & 100.0 & {\it   4.6} &  22.7 &  86.9 \\ 
   & 400 & 100.0 &  96.2 & 100.0 & {\it   4.1} &  22.7 &  98.5 & 100.0 &  97.2 & 100.0 & {\it   5.4} &  30.4 &  96.6 \\ 
   & 500 & 100.0 &  98.9 & 100.0 & {\it   5.0} &  29.3 &  99.1 & 100.0 &  99.0 & 100.0 & {\it   6.2} &  37.8 &  99.5 \\ 
   [1ex]
T5 & 100 &  82.2 &  30.0 &  75.1 &   3.9 & {\it   1.6} &  26.9 &  83.3 &  29.3 &  97.7 &   4.9 & {\it   2.0} &  24.9 \\ 
   & 200 &  98.8 &  75.4 &  99.3 &   6.3 & {\it   3.1} &  78.3 &  98.1 &  75.5 & 100.0 &   7.0 & {\it   2.8} &  76.4 \\ 
   & 300 &  99.8 &  92.8 &  99.8 &   7.8 & {\it   3.5} &  95.6 &  99.9 &  91.3 & 100.0 &   8.2 & {\it   4.9} &  92.2 \\ 
   & 400 & 100.0 &  98.8 & 100.0 &   8.9 & {\it   4.3} &  99.6 & 100.0 &  97.2 & 100.0 &   8.2 & {\it   5.4} &  98.4 \\ 
   & 500 &  99.9 &  99.9 & 100.0 &  11.3 & {\it   5.1} & 100.0 & 100.0 &  99.1 & 100.0 &  12.6 & {\it   6.5} &  99.7 \\ 
   [1ex]
GN & 100 &  15.9 &   7.5 &  38.4 &   4.5 &   5.1 & {\it   2.3} &  14.1 &   5.0 &  92.7 &   4.0 &   4.6 & {\it   1.8} \\ 
   & 200 &  30.0 &  17.6 &  79.7 &  22.2 &  21.0 & {\it   4.0} &  25.8 &  14.8 &  99.9 &  25.0 &  20.1 & {\it   4.2} \\ 
   & 300 &  40.5 &  26.4 &  93.3 &  46.0 &  40.6 & {\it   3.3} &  42.5 &  23.8 & 100.0 &  48.3 &  38.1 & {\it   4.8} \\ 
   & 400 &  52.9 &  32.3 &  99.2 &  64.4 &  57.5 & {\it   4.5} &  54.6 &  25.4 & 100.0 &  70.3 &  52.0 & {\it   3.7} \\ 
   & 500 &  63.8 &  37.6 &  99.5 &  81.2 &  74.2 & {\it   4.3} &  64.9 &  32.8 & 100.0 &  82.0 &  69.3 & {\it   4.4} \\ 
  \bottomrule
\end{tabular}
\end{center}
\end{table}

\subsection{The case of the multivariate $t$ distribution} 

When $\MM$ is the set of c.d.f.s from the multivariate $t$ distribution with fixed d.f., two different ways of computing the gradients $\dot F_\theta$ and $\dot f_\theta / f_\theta$, $\theta \in \OO$, were considered. The first one is the generic approach mentioned earlier relying on Richardson's extrapolation method as implemented in the R package {\tt numDeriv}. This numerical method requires numerous evaluations of the c.d.f.\ and the log p.d.f.\ of the multivariate $t$ as it is based on finite differences. For that reason, in the trivariate experiments, the {\tt algorithm} parameter of the {\tt pmvt} and {\tt pmvnorm} functions of the {\tt mvtnorm} package was set to {\tt TVPACK} \citep{Gen04} instead of the default algorithm {\tt GenzBretz}. Indeed, the latter is based on randomized quasi Monte Carlo methods \citep{GenBre99,GenBre02} which implies that its results depend on random number generation and are therefore not fully reproducible. The second approach, which is expected to be more precise, is based, as explained in Appendix~\ref{mult_t}, on the fact that the gradient $\dot f_{\theta} / f_{\theta}$ can be computed analytically, and on the fact that the gradient $\dot F_\theta$ can be expressed in terms of the c.d.f.\ and the p.d.f.\ of the multivariate $t$. 

These two approaches were thoroughly compared and their results were found to be very close. As could have been expected, the second more analytical approach is much faster as it requires significantly fewer evaluations of the c.d.f.\ and the p.d.f.\ of the multivariate $t$. This aspect will be illustrated in the next section.

Let us finally discuss the estimation of the parameters of the multivariate $t$ with fixed d.f.\ $\nu$ using the Nelder-Mead algorithm. In dimension three for small $\nu \geq 3$, we noticed that the estimation of the nine parameters was extremely sensitive to the choice of the starting values, and that the multiplier test tended to be too liberal. Such issues were not observed for the other families of c.d.f.s used in the trivariate simulations. Improved results were obtained by changing the scale of the optimization using the {\tt parscale} argument of the R {\tt optim} routine. The latter argument was set to the vector of starting values with a guard against values too close to zero. Additional simulations were carried out for $\nu=3$, 5, 7, 10, 20 and 30 to study the empirical levels of the multiplier test. Table~\ref{sizes} gives such empirical levels when data are generated from bivariate and trivariate $t$ distributions with $\nu$ d.f.\ fitted from the financial data studied in Section~\ref{illustration}. As one can notice, in dimension two, the empirical levels are reasonably close to the 5\% significance level. In dimension three however, the test is clearly too liberal for $\nu=3$ and might be slightly too liberal for $\nu=5$. We could not determine whether a similar issue also affects the parametric bootstrap for computational reasons. We do however believe that the use of estimation procedures specifically tailored to the multivariate $t$ \citep[see e.g.][Section~3]{NadKot08} would solve this problem.  

\begin{table}[t!]
\caption{Rejection rate (in \%) of the null hypothesis that the data come from a $d$-dimensional $t$ distribution with fixed d.f.\ $\nu$ as observed in 1000 random samples of size $n=200$, 500, 1000 and 2000 generated from bivariate and trivariate $t$ distributions with $\nu$ d.f.\ fitted from the financial data used in Section~\ref{illustration}.}
\label{sizes}
\begin{center}
\begin{tabular}{r rrrrrr rrrrrr}
  \toprule
\multicolumn{1}{l}{$n$} & \multicolumn{6}{c}{$d = 2$} & \multicolumn{6}{c}{$d = 3$}\\
\cmidrule(lr){2-7}\cmidrule(lr){8-13}
$\nu$    &  3 & 5 & 7 & 10 & 20 & 30 & 3 & 5 & 7 & 10 & 20 & 30\\
  \midrule
200 & 4.2 & 3.4 & 3.1 & 3.5 & 4.1 & 4.4 & 8.6 & 3.2 & 3.7 & 2.2 & 2.3 & 3.3 \\ 
  500 & 4.7 & 4.5 & 5.6 & 4.7 & 4.7 & 4.1 & 11.8 & 5.0 & 3.7 & 4.5 & 4.0 & 3.5 \\ 
  1000 & 4.6 & 4.4 & 4.6 & 5.6 & 5.6 & 5.0 & 12.9 & 7.6 & 4.6 & 5.1 & 4.3 & 4.7 \\ 
  2000 & 6.1 & 4.2 & 5.2 & 5.6 & 4.5 & 4.8 & 11.1 & 6.6 & 4.8 & 5.2 & 4.9 & 4.8 \\ 
  \bottomrule
\end{tabular}
\end{center}
\end{table}

\section{Illustration}
\label{illustration}

The results of the univariate, bivariate and trivariate experiments whose results were partially reported in the previous section hence suggest that the multiplier procedure can be safely used as a large-sample alternative to the parametric bootstrap. 

To illustrate the computational advantage of the multiplier procedure over the parametric bootstrap, we consider the financial data analyzed in \citet[Chapter 5]{McNFreEmb05}. These consist of five years of daily log-returns (1996-2000) for the Intel (INTC), Microsoft (MSFT) and General Electric (GE) stocks, which gives a trivariate sample of size $n=1262$. Univariate goodness-of-fit tests for N, T5, T10, T20 and L were first applied to the Intel log-return data. The other distributions used in the univariate simulations were not considered as their support is $(0,\infty)$. Approximate $p$-values and execution times are reported in the first horizontal block of Table~\ref{illus}. The execution times are obtained from our R implementations of the multiplier procedure and of the parametric bootstrap. As one can see, the approximate $p$-values of the multiplier procedure and the parametric bootstrap are fairly close, and the execution times are of the same order of magnitude. The latter observation is essentially due to the fact that (i)~the numerical estimation of the parameters of the hypothesized distributions (on which the parametric bootstrap heavily relies) is reasonably fast in dimension one, and (ii)~the gradients needed in the multiplier procedure are computed numerically using the {\tt numDeriv} package as explained previously. Bivariate goodness-of-fit tests for N, NC, T10N, T5, T10 and T20 were then applied to the bivariate log-returns of the Intel and General Electric stocks. Again, the approximate $p$-values of the multiplier procedure and the parametric bootstrap are fairly close. This time, however, the computational advantage of the multiplier procedure is obvious, in particular when T10N is hypothesized (1.6 minutes versus 4.2 hours for the parametric bootstrap on one 2.33 GHz processor). The approximate $p$-values and executions times in italic in Table~\ref{illus} were obtained using the multiplier procedure based on the gradients computed using the more analytical approach described in Appendix~\ref{mult_t}. Finally, goodness-of-fit tests for N, NC, T10N, T5, T10 and T20 were applied to the trivariate log-returns. From the third horizontal block of Table~\ref{illus}, we see that the computational advantage of the multiplier is even more pronounced than in dimension two. For T10N for instance, the execution of the multiplier procedure took 4.3 minutes while 16.6 hours were necessary to obtain an approximate $p$-value using the parametric bootstrap.

\begin{table}[t!]
\caption{Approximate $p$-values and executions times on one 2.33 GHz processor of the multiplier procedure (MP) and the parametric bootstrap (PB) for the financial data analyzed in \citet[Chapter 5]{McNFreEmb05}. The approximate $p$-values and executions times in italic were obtained using MP in which the gradients are computed using the more analytical approach described in Appendix~\ref{mult_t}.}
\label{illus}
\begin{center}
\begin{tabular}{lr rrrr}
  \toprule
Variables & Distribution & \multicolumn{2}{c}{$p$-value} & \multicolumn{2}{c}{Time in seconds}\\
\cmidrule(lr){3-4}\cmidrule(lr){5-6}
    &     &  MP & PB & MP & PB\\
  \midrule
  INTC & N & 0.000 & 0.000 & 9.5 & 14.0 \\ 
  & T5 & 0.066 & 0.077 & 9.3 & 43.2 \\ 
  & T10 & 0.538 & 0.520 & 9.3 & 40.1 \\ 
  & T20 & 0.034 & 0.017 & 9.3 & 40.4 \\ 
  & L & 0.461 & 0.405 & 9.2 & 18.7 \\
  [1ex] 
  (INTC, GE) & N & 0.000 & 0.000 & 79.2 & 1934.9 \\ 
  & NC & 0.000 & 0.000 & 35.6 & 12185.0 \\ 
  & T10N & 0.022 & 0.012 & 98.3 & 14969.8 \\ 
  & T5 & {\it 0.043} & 0.026 & {\it 8.0} & 2089.8 \\ 
  & T10 & {\it 0.187} & 0.184 & {\it 9.0} & 2093.1 \\  
  & T10 & 0.200 & 0.181 & 95.6 & 2238.1 \\ 
  & T20 & {\it 0.003} & 0.004 & {\it 9.2} & 2088.0 \\ 
  [1ex]
  (INTC, GE, MSFT) & N & 0.000 & 0.000 & 166.1 & 3578.3 \\ 
  & NC & 0.000 & 0.000 & 93.9 & 27954.9 \\ 
  & T10N & 0.000 & 0.000 & 256.0 & 59902.1 \\ 
  & T5 & {\it 0.077} & 0.097 & {\it 15.5} & 4340.0 \\ 
  & T10 & {\it 0.119} & 0.139 & {\it 14.9} & 4459.4 \\
  & T10 & 0.133 & 0.149 & 252.8 & 4740.8 \\  
  & T20 & {\it 0.004} & 0.021 & {\it 14.4} & 4476.4 \\
\bottomrule
\end{tabular}
\end{center}
\end{table}

The entries in italic in Table~\ref{illus} show that the execution times of the generic implementation of the multiplier procedure based on the {\tt numDeriv} package can be significantly lowered at the expense of more analytical and programming work. Formulas similar to or simpler than those given in Appendix~\ref{mult_t} could be obtained for all the multivariate distributions considered in this work.

The approximate $p$-values given in the third horizontal block of Table~\ref{illus} indicate that there is very little evidence against the trivariate distributions T5 and T10. On the basis of these tests, we can conclude that a trivariate $t$ distribution whose d.f.\ are close to 10 is a plausible model for these financial data. 

\section*{Acknowledgments}

The authors would like to thank the associate editor and a referee for their insightful and constructive comments as well as Laurent Bordes for fruitful discussions.

\appendix

\section{Proofs of the propositions}
\label{proofs}

The following lemma will be used in the proofs of Propositions~\ref{propH0} and~\ref{propH1}.

\begin{lem}
\label{lem1}
Let $X_1,\dots,X_n$ be a random sample from a distribution $P$ that may or may not belong to $\{P_\theta : \theta \in \OO\}$. If Assumptions~(A1)-(A4) are satisfied, and if there exists $\theta_0 \in \OO$ such that $\theta_n$ converges in probability to $\theta_0$ under $P$, then
$$
\left(f \mapsto  \Gn f - \Gn \psi_{\theta_0}^\top \dot P_{\theta_0}f, f \mapsto \Gn'f - \Gn' \psi_{\theta_n}^\top \dot P_{\theta_n}f, f \mapsto \Gn''f - \Gn'' \psi_{\theta_n}^\top \dot P_{\theta_n}f \right) 
$$
converges weakly to 
$$
\left(f \mapsto \G_Pf - \G_P \psi_{\theta_0}^\top \dot P_{\theta_0}f , f \mapsto \G_P'f - \G_P' \psi_{\theta_0}^\top \dot P_{\theta_0}f, f \mapsto \G_P'f - \G_P' \psi_{\theta_0}^\top \dot P_{\theta_0}f \right)
$$
in $\{ \ell^\infty(\FF) \}^3$, where $\G_P$, a $P$-Brownian bridge, is the weak limit of $\Gn$, and $\G_P'$ is an independent copy of $\G_P$.
\end{lem}

\begin{proof}
From Assumption (A3), $\{\psi_{\theta_0}\}$ is $P$-Donsker. It follows that the class $\GG$ obtained as the union of $\FF$ and the $p$ components of $\psi_{\theta_0}$ is $P$-Donsker. From the functional multiplier central limit theorem \citep[see e.g.][Theorem 10.1 and Corollary 10.3]{Kos08}, we then have that
$$
\left(\Gn,\Gn',\Gn'' \right) \leadsto \left(\G_P,\G_P',\G_P' \right)
$$
in $\{ \ell^\infty(\GG) \}^3$. By the continuous mapping theorem, it first follows that
$$
\left( \Gn, \Gn\psi_{\theta_0}, \Gn, \Gn'\psi_{\theta_0}, \Gn, \Gn''\psi_{\theta_0} \right) \leadsto \left( \G_{\theta_0}, \G_{\theta_0}\psi_{\theta_0}, \G_{\theta_0}', \G_{\theta_0}'\psi_{\theta_0}, \G_{\theta_0}', \G_{\theta_0}'\psi_{\theta_0} \right)
$$
in $\{ \ell^\infty(\FF) \times \Rset^p \}^3$, and then that
\begin{equation}
\label{conv1}
\left(f \mapsto  \Gn f - \Gn \psi_{\theta_0}^\top \dot P_{\theta_0}f, f \mapsto \Gn'f - \Gn' \psi_{\theta_0}^\top \dot P_{\theta_0}f, f \mapsto \Gn''f - \Gn'' \psi_{\theta_0}^\top \dot P_{\theta_0}f \right) 
\end{equation}
converges weakly to 
\begin{equation}
\label{conv2}
\left(f \mapsto \G_Pf - \G_P \psi_{\theta_0}^\top \dot P_{\theta_0}f , f \mapsto \G_P'f - \G_P' \psi_{\theta_0}^\top \dot P_{\theta_0}f, f \mapsto \G_P'f - \G_P' \psi_{\theta_0}^\top \dot P_{\theta_0}f \right)
\end{equation}
in $\{ \ell^\infty(\FF) \}^3$. 

Now, from Assumption (A3), there exists a $\delta > 0$ such that $\HH = \{\psi_\theta : \|\theta - \theta_0 \| < \delta \}$ is $P$-Donsker. Let $\HH_k$, $k \in \{1,\dots,p\}$, be the $p$ component classes of $\HH$. They are $P$-Donsker by definition. 

Next, fix $k \in \{1,\dots,p\}$ and let $g$ be a function from $\ell^\infty(\HH_k) \times \HH_k \to \Rset$ defined by $g(z,\psi) = z(\psi) - z(\psi_{\theta_0,k})$, where $\psi_{\theta_0,k}$ is the $k$th component of $\psi_{\theta_0}$. As noted in \citet[proof of Lemma 19.24]{van98}, the set $\HH_k$ is a semimetric space with respect to metric $L_2(P)$ and the function $g$ is continuous with respect to the product semimetric on $\ell^\infty(\HH_k) \times \HH_k$ at every point $(z,\psi)$ such that $\psi \mapsto z(\psi)$ is continuous.

From Assumption (A4), the fact that $\theta_n$ converges to $\theta_0$ in probability, and Lemma 2.12 of \citet{van98}, we also have that $\psi_{\theta_n,k}$ converges in probability to $\psi_{\theta_0,k}$ in the space $\HH_k$ equipped with the metric $L_2(P)$. Since $\HH_k$ is $P$-Donsker, $\Gn' \leadsto \G_P'$ in $\ell^\infty(\HH_k)$. Also, since $\theta_n$ converges to $\theta_0$ in probability, the probability that $\psi_{\theta_n,k}$ is in $\HH_k$ tends to 1. On that event, it follows that $(\Gn',\psi_{\theta_n,k}) \leadsto (\G_P',\psi_{\theta_0,k})$ in $\ell^\infty(\HH_k) \times \HH_k$. Since $\psi_k \mapsto \G_P' \psi_k $ is continuous at every $\psi_k \in \HH_k$ almost surely, the function $g$ is continuous at almost every $(\G_P',\psi_{\theta_0,k})$. By the continuous mapping theorem, we obtain that $g(\Gn',\psi_{\theta_n,k}) = \Gn'\psi_{\theta_n,k} - \Gn'\psi_{\theta_0,k} \leadsto g(\G_P',\psi_{\theta_0,k}) = 0$. Hence, we have that 
\begin{equation}
\label{asym_equi2}
\Gn'\psi_{\theta_n,k} - \Gn'\psi_{\theta_0,k} = o_P (1) , \qquad k \in \{1,\dots,p\}.
\end{equation}
Similarly, we have that
\begin{equation}
\label{asym_equi3}
\Gn''\psi_{\theta_n,k} - \Gn''\psi_{\theta_0,k} = o_P (1), \qquad k \in \{1,\dots,p\}.
\end{equation}
From Assumption (A2) and the fact $\theta_n$ converges to $\theta$ in probability, we also have that 
\begin{equation}
\label{asym_equi4}
\sup_{f \in \FF} \| \dot P_{\theta_n} f - \dot P_{\theta_0} f \| = o_P (1).
\end{equation}
Finally, combining~(\ref{conv1}) and~(\ref{conv2}) with~(\ref{asym_equi2}),~(\ref{asym_equi3}) and~(\ref{asym_equi4}), we obtain the desired result.
\end{proof}


We can now prove Proposition~\ref{propH0} and~\ref{propH1}.

\begin{proof}[\bf Proof of Proposition~\ref{propH0}] Assumption (A5) implies that $\theta_n$ converges to $\theta_0$ in probability. From Assumption (A1) and Lemma 2.12 of \citet{van98}, we then have that
$$
 \sup_{f \in \FF} | P_{\theta_n}f - P_{\theta_0}f - (\theta_n - \theta_0)^\top \dot P_{\theta_0}f | = \| \theta_n - \theta_0 \| o_{ P_{\theta_0}} (1),
$$
which in turn is implies that
$$
 \sup_{f \in \FF} | \sqrt{n} (P_{\theta_n}f - P_{\theta_0}f) -\sqrt{n} (\theta_n - \theta_0)^\top \dot P_{\theta_0}f | =  o_{ P_{\theta_0}} (1),
$$
since $\| \sqrt{n}  (\theta_n - \theta_0) \| =  O_{ P_{\theta_0}} (1)$ from Assumption (A5) and the continuous mapping theorem. It follows that
\begin{align*}
\sqrt{n} (\Pn - P_{\theta_n}) &= \sqrt{n} (\Pn - P_{\theta_0}) - \sqrt{n} (P_{\theta_n} - P_{\theta_0}) \\
&= \sqrt{n} (\Pn - P_{\theta_0}) - \sqrt{n} (\theta_n - \theta_0)^\top \dot P_{\theta_0} + R_n,
\end{align*}
where $\sup_{f \in \FF} |R_n f| =  o_{ P_{\theta_0}} (1)$. Using Assumption (A5) again, we obtain that
\begin{equation}
\label{asym_equi1}
\sqrt{n} (\Pn f - P_{\theta_n} f) = \Gn f - \Gn \psi_{\theta_0}^\top \dot P_{\theta_0}f + Q_n f, \qquad f \in \FF,
\end{equation}
where $\sup_{f \in \FF} |Q_n f| =  o_{ P_{\theta_0}} (1)$. The result is finally an immediate consequence of Lemma~\ref{lem1}.
\end{proof}


\begin{proof}[\bf Proof of Proposition~\ref{propH1}]
Write
\begin{equation}
\label{decomposition}
\sqrt{n} (\Pn - P_{\theta_n}) = \sqrt{n} (\Pn - P) -  \sqrt{n} (P_{\theta_n} - P_{\theta_0}) - \sqrt{n} (P_{\theta_0} - P).
\end{equation}
The first term converges weakly to $\G_P$ in $\ell^\infty(\FF)$. Now, the convergence in distribution of $\sqrt{n} (\theta_n - \theta_0)$ implies that $\theta_n$ converges in probability to $\theta_0$. Hence, proceeding as in the proof of Proposition~\ref{propH0}, Assumption (A1) implies that
$$
 \sup_{f \in \FF} | \sqrt{n} (P_{\theta_n}f - P_{\theta_0}f) -\sqrt{n} (\theta_n - \theta_0)^\top \dot P_{\theta_0}f | =  o_P (1),
$$
which in turn implies that the second term in~(\ref{decomposition}) converges weakly in $\ell^\infty(\FF)$. However, since $P \not \in \{P_\theta : \theta \in \OO\}$, the supremum over $\FF$ of the third term diverges, which implies that  
$$
\sup_{f \in \FF} | \sqrt{n} (\Pn f - P_{\theta_n} f) | \overset{P}{\to} \infty.
$$

The second part of the proposition is an immediate consequence of Lemma~\ref{lem1}.
\end{proof}

\section{Computational details for the multivariate $t$}
\label{mult_t}

The p.d.f.\ of the centered $d$-dimensional multivariate $t$ with dispersion matrix $\Sigma$ and $\nu$ d.f.\ is given by
\begin{equation}
\label{pdf_t}
t_{\nu,\Sigma}(x) =  \frac{\Gamma \left( \frac{\nu+d}{2} \right)}{(\pi \nu)^{\frac{d}{2}} \Gamma\left(\frac{\nu}{2}\right) |\Sigma|^{\frac{1}{2}}} \left( 1 + \frac{1}{\nu} x^\top \Sigma^{-1} x\right)^{-\frac{\nu+d}{2}}, \qquad x \in \Rset^d.
\end{equation}
Let $T_{\nu,\Sigma}$ denote the corresponding c.d.f. It is easy to verify that the c.d.f.\ of the multivariate $t$ with $\nu$ d.f., expectation vector $(\mu_1,\dots,\mu_d)$, dispersions $\lambda_1,\dots,\lambda_d$ and correlation matrix $\Sigma$ is then given by 
$$
T_{\nu,\Sigma,\mu,\lambda}(x) = T_{\nu,\Sigma} \left( \frac{x_1-\mu_1}{\lambda_1},\dots, \frac{x_d-\mu_d}{\lambda_d} \right), \qquad x \in \Rset^d.
$$
The corresponding p.d.f.\ is thus
$$
t_{\nu,\Sigma,\mu,\lambda}(x) = \left( \prod_{j=1}^d \lambda_j \right)^{-1} t_{\nu,\Sigma} \left( \frac{x_1-\mu_1}{\lambda_1},\dots, \frac{x_d-\mu_d}{\lambda_d} \right), \qquad x \in \Rset^d.
$$

Let us first explain how, for any $x \in \Rset^d$, the gradient of $T_{\nu,\Sigma,\mu,\lambda}(x)$ with respect to all the parameters except $\nu$ can be computed. Let $j \in \{1,\dots,d\}$, and, for any $x \in \Rset^d$, let $T_{\nu,\Sigma}^{(j)}(x) = \partial T_{\nu,\Sigma}(x) / \partial x_j$. Also, let $\Sigma_{-j,-j}$ be a $(d-1) \times (d-1)$ matrix obtained from $\Sigma$ by removing its $j$th row and $j$th column, $\Sigma_{-j,j}$ be a $(d-1) \times 1$ matrix obtained from $\Sigma$ by removing its $j$th row and keeping only its $j$th column, and $\Sigma_{j,-j} = \Sigma_{-j,j}^\top $. From \citet[page 66]{NadKot05}, if $X$ is standard multivariate $t$ with $\nu$ d.f.\ and correlation matrix $\Sigma$, then, conditionally on $X_j=x_j$, we have that 
$$
\sqrt{\frac{\nu+1}{\nu+x_j^2}} \left( X_{-j} - x_j \Sigma_{-j,j}\right)
$$
is multivariate centered $t$ with $\nu+1$ degrees of freedom and dispersion matrix $\Lambda_j = \Sigma_{-j,-j} - \Sigma_{-j,j} \Sigma_{j,-j}$. Hence,
$$
T_{\nu,\Sigma}^{(j)}(x)= t_\nu(x_j) \, T_{\nu+1,\Lambda_j} \left( \sqrt{\frac{\nu+1}{\nu+x_j^2}} (x_{-j} - x_j \Sigma_{-j,j}) \right), \qquad x \in \Rset^d.
$$
Using the previous expression, it is therefore possible to compute
$$
\frac{\partial T_{\nu,\Sigma,\mu,\lambda}(x)}{\partial \mu_j} = - \lambda_j^{-1}T_{\nu,\Sigma}^{(j)}\left( \frac{x_1-\mu_1}{\lambda_1},\dots, \frac{x_d-\mu_d}{\lambda_d} \right), \qquad x \in \Rset^d,
$$
and 
$$
\frac{\partial T_{\nu,\Sigma,\mu,\lambda}(x)}{\partial \lambda_j^2} = - \frac{x_j - \mu_j}{2 \lambda_j^3} T_{\nu,\Sigma}^{(j)}\left( \frac{x_1-\mu_1}{\lambda_1},\dots, \frac{x_d-\mu_d}{\lambda_d} \right), \qquad x \in \Rset^d.
$$
Also, let $\rho_{i,j}$ be an off-diagonal element of the correlation matrix $\Sigma$. Then,
$$
\frac{\partial T_{\nu,\Sigma,\mu,\lambda}(x)}{\partial \rho_{i,j}} = \frac{\partial T_{\nu,\Sigma}\left( \frac{x_1-\mu_1}{\lambda_1},\dots, \frac{x_d-\mu_d}{\lambda_d} \right)}{\partial \rho_{i,j}}, \qquad x \in \Rset^d,
$$
and, for any $x \in \Rset^d$, $\partial T_{\nu,\Sigma}(x) / \partial \rho_{i,j}$ can be computed using the Plackett formula for the multivariate $t$ \citep[see][Proposition~1]{Gen04,KojYan11}. 

Let us now discuss the computation, for any $x \in \Rset^d$, of the gradient of $\log t_{\nu,\Sigma,\mu,\lambda}(x)$ with respect to all the parameters except $\nu$. Let $j \in \{1,\dots,d\}$, and, for any $x \in \Rset^d$, let $t_{\nu,\Sigma}^{(j)}(x) = \partial t_{\nu,\Sigma}(x) / \partial x_j$. Starting from~(\ref{pdf_t}), one obtains that 
$$
t^{(j)}_{\nu,\Sigma}(x) =  - \frac{(\nu+d) x^\top \Sigma^{-1} e_j}{\nu + x^\top \Sigma^{-1} x} t_{\nu, \Sigma}(x), \qquad x \in \Rset^d, 
$$
where $e_j$ is the unit vector of $\Rset^d$ whose $i$th component is 1 if $i=j$ and 0 otherwise. Using the previous expression, it is therefore possible to compute
$$
\frac{\partial \log t_{\nu,\Sigma,\mu,\lambda}(x)}{\partial \mu_j} = - \lambda_j^{-1} \frac{t_{\nu,\Sigma}^{(j)} \left( \frac{x_1-\mu_1}{\lambda_1},\dots, \frac{x_d-\mu_d}{\lambda_d} \right)}{t_{\nu,\Sigma} \left( \frac{x_1-\mu_1}{\lambda_1},\dots, \frac{x_d-\mu_d}{\lambda_d} \right)}, \qquad x \in \Rset^d,
$$
and 
$$
\frac{\partial \log t_{\nu,\Sigma,\mu,\lambda}(x)}{\partial \lambda_j^2} = - \frac{1}{2 \lambda_j^2} - \frac{x_j - \mu_j}{2 \lambda_j^3} \frac{t_{\nu,\Sigma}^{(j)} \left( \frac{x_1-\mu_1}{\lambda_1},\dots, \frac{x_d-\mu_d}{\lambda_d} \right)}{t_{\nu,\Sigma} \left( \frac{x_1-\mu_1}{\lambda_1},\dots, \frac{x_d-\mu_d}{\lambda_d} \right)}, \qquad x \in \Rset^d.
$$
Finally, starting again from~(\ref{pdf_t}), one obtains that
$$
\frac{\partial t_{\nu, \Sigma}(x)}{\partial \rho_{i,j}}  = - \frac{1}{2} t_{\nu, \Sigma}(x) \left\{ \frac{1}{| \Sigma|} \frac{\partial | \Sigma|}{\partial \rho_{i,j}} + \frac{(\nu + d) x^\top \frac{\partial  \Sigma^{-1}}{\partial \rho_{i,j}} x}{\nu + x^\top  \Sigma^{-1} x} \right\}, \qquad x \in \Rset^d.
$$
From \citet[Chap. 17]{Seb08} for instance, we have that
$$
\frac{\partial | \Sigma|}{\partial \rho_{i,j}} = 2 K_{ij}
\qquad \mbox{and that} \qquad
\frac{\partial  \Sigma^{-1}}{\partial \rho_{i,j}} = - r_i r_j^\top - r_j r_i^\top,
$$
where $K_{ij}$ is the cofactor of $\rho_{i,j}$, and where $r_i$ is the $i$-th column of $\Sigma^{-1}$.

\bibliographystyle{plainnat}
\bibliography{biblio}

\end{document}